\documentclass[12pt,a4paper,english]{article}
\usepackage[ansinew]{inputenc}
\usepackage[T1]{fontenc}
\usepackage[english]{babel}
\usepackage{fancybox}
\usepackage{fancyhdr}
\usepackage{color}
\usepackage{setspace}
\newcommand{\sets}{\setstretch{1.1}}
\usepackage{amsmath}
\usepackage{amstext}
\usepackage{amssymb}
\usepackage{amsthm}
\usepackage{mathrsfs}
\newcommand{\be}{\begin{equation}}
\newcommand{\ee}{\end{equation}}
\newcommand{\ba}{\begin{aligned}}
\newcommand{\ea}{\end{aligned}}
\newcommand{\R}{\mathbb{R}}
\newcommand{\N}{\mathbb{N}}
\newcommand{\ind}{\mathbf{1}}
\newcommand{\lsi}{\left[\negthinspace\left[}
\newcommand{\rsi}{\right]\negthinspace\right]}

\newcommand{\A}{\mathcal{A}}
\newcommand{\B}{\mathcal{B}}
\newcommand{\cP}{\mathcal{P}}
\newcommand{\M}{\mathcal{M}}
\newcommand{\F}{\mathcal{F}}
\newcommand{\FF}{\mathbb{F}}
\newcommand{\G}{\mathcal{G}}
\newcommand{\GG}{\mathbb{G}}
\newcommand{\calS}{\mathcal{S}}
\newcommand{\tN}{\widetilde{N}}
\newcommand{\bit}{\bibitem}
\renewcommand{\phi}{\varphi}

\newtheorem{Thm}{\bf Theorem}[section]
\newtheorem{Def}[Thm]{\bf Definition}
\newtheorem{Prop}[Thm]{\bf Proposition}
\newtheorem{Lem}[Thm]{\bf Lemma}
\newtheorem{Cor}[Thm]{\bf Corollary}
\theoremstyle{remark}
\newtheorem{Rem}[Thm]{\bf Remark}

\newtheorem{Pb}[Thm]{\bf Problem}

\newtheorem{Ass}[Thm]{\bf Assumption}

\newenvironment{myenumerate}{%

\begin{enumerate}}{\end{enumerate}}
\numberwithin{equation}{section}
\usepackage{geometry}
\makeatletter
\renewcommand*\@fnsymbol[1]{\the#1}
\makeatother

\title{Market viability and martingale measures under partial information}

\author{Claudio Fontana\thanks{INRIA Paris-Rocquencourt, Domaine de Voluceau, Rocquencourt, BP 105, Le Chesnay Cedex, 78153, France, e-mail: {\tt claudio.fontana@inria.fr}.}
\and
Bernt \O ksendal\thanks{Center of Mathematics for Applications (CMA), Dept. of Mathematics, University of Oslo, P.O. Box 1053 Blindern, N--0316 Oslo, Norway, e-mail: {\tt oksendal@math.uio.no}.}
\and
Agn\`es Sulem\thanks{INRIA Paris-Rocquencourt, Domaine de Voluceau, Rocquencourt, BP 105, Le Chesnay Cedex, 78153, France, e-mail: {\tt agnes.sulem@inria.fr}.}}

\date{This version: October 15, 2013}

\begin{document}

\maketitle

\abstract{\begin{spacing}{1}\noindent We consider a financial market model with a single risky asset whose price process evolves according to a general jump-diffusion with locally bounded coefficients and where market participants have only access to a partial information flow. For any utility function, we prove that the partial information financial market is locally viable, in the sense that the optimal portfolio problem has a solution up to a stopping time, if and only if the (normalised) marginal utility of the terminal wealth generates a partial information equivalent martingale measure (PIEMM). This equivalence result is proved in a constructive way by relying on maximum principles for stochastic control problems under partial information. We then characterize a global notion of market viability in terms of partial information local martingale deflators (PILMDs). We illustrate our results by means of a simple example.\end{spacing}}
\vspace{0.5cm}

\begin{spacing}{1}\noindent 
\end{spacing}\vspace{0.3cm}


\section{Introduction}	\label{S1}

The concepts of no-arbitrage, martingale measure and optimal portfolio can be rightly considered as the cornerstones of modern mathematical finance, starting from the seminal papers \cite{HK,K}. Loosely speaking, the no-arbitrage requirement is equivalent to the existence of a martingale measure, which can then be used for pricing purposes (risk-neutral valuation), and, again loosely speaking, portfolio optimization problems are solvable if and only if arbitrage profits cannot be obtained by trading on the market. 

In the context of discrete-time models on a finite probability space (see e.g. \cite{PR,P}), it can actually be shown that the above concepts are equivalent and, furthermore, one can work out explicitly the connections between them. 
In particular, portfolio optimization problems and solvable if and only if there exists an equivalent martingale measure (EMM) and, moreover, one can obtain an EMM by taking the (normalised) marginal utility of the optimal terminal wealth. This relation also represents a classical and well-known result from the economic literature (see e.g. \cite{B}, Section 4.4). In the case of discrete-time models on a general probability space, the validity of this equivalence has been studied in \cite{RS,Sch1,Sch2}.

When one moves from discrete-time to continuous-time financial models, then things become quickly more complicated and the equivalences discussed so far do not hold any more in full generality. For instance, in order to recover the equivalence between EMMs and no-arbitrage, one has to replace the notion of martingale with the notion of local martingale (or $\sigma$-martingale) and the condition of no-arbitrage with the \emph{no free lunch with vanishing risk} (NFLVR) criterion adopted in \cite{DS0,DS2} (see also \cite{Frit} for an equivalent characterization of NFLVR in terms of \emph{no market free lunch}, a no-arbitrage criterion based on the structure of investors' preferences). Furthermore, the marginal utility of the optimal terminal wealth does not necessarily yield the density of an equivalent (local-/$\sigma$-)martingale measure but only the terminal value of a \emph{supermartingale deflator} (see e.g. \cite{KS,S}).

In the present paper, we consider a general jump-diffusion model with locally bounded coefficients and study the issue of the \emph{viability} of the financial market, defined as the ability to solve a portfolio optimization problem. Our main goal consists in characterizing the notion of viability in terms of martingale measures, in a sense to be made precise in the following, studying under which conditions the marginal utility of terminal wealth gives rise to a martingale measure. We refrain from a-priori imposing no-arbitrage restrictions on the model, tackling instead directly the solvability of portfolio optimization problems. Furthermore, we suppose that market participants have only access to a partial information flow, which does not reveal the full information of the stochastic basis.
In order to solve portfolio optimization problems under partial information, we shall employ necessary and sufficient maximum principles for stochastic control problems under partial information, as discussed in \cite{BO} (see also the recent paper \cite{OS2} for related results in the complete information case). This approach allows us to characterize the optimal solution via an associated BSDE, which in turn requires a good control on the integrability properties of the processes involved. Since such integrability conditions are not satisfied in general, we need to resort to a localization procedure, as explained in Section \ref{S3}.

The main contributions of the present paper can be outlined as follows:
\begin{myenumerate}
\item
we show that the financial market under partial information is \emph{locally viable}, in the sense that a portfolio optimization problem admits a solution up to a stopping time, if and only if there exists a \emph{partial information equivalent martingale measure} (PIEMM) up to a stopping time. Furthermore, the density of such PIEMM is given by the (normalised) marginal utility of the optimal terminal wealth, thus recovering the classical result of financial economics;
\item
we prove that, if the financial market under partial information is \emph{globally viable}, in the sense that it is locally viable for a sequence $\{\tau_n\}_{n\in\N}$ of increasing stopping times, then there exists a \emph{partial information local martingale deflator} (PILMD). Furthermore, we show that such PILMD can be constructed by aggregating the densities of all local PIEMMs if and only if the locally optimal portfolios satisfy a consistency condition;
\item
as a special case, if the price process has bounded coefficients, we prove that the financial market is viable on the global time horizon if and only if the (normalised) marginal utility of the optimal terminal wealth defines a PIEMM on the global time horizon;
\item
by means of a simple and classical example (see Section \ref{S5}) we show that, even for regular utility functions and continuous-path processes with good integrability properties but unbounded coefficients, a PIEMM may fail to exist globally but, nevertheless, a PILMD exists.
\end{myenumerate}

To the best of our knowledge, the issue of linking the viability of the financial market to the existence of weaker counterparts of equivalent martingale measures such as PIEMMs and PILMDs has never been dealt with in the partial information case. Furthermore, we go significantly beyond a pure existence result, since our approach allows us to obtain a precise and explicit connection between the solution to an optimal portfolio problem and the density of an equivalent martingale measure / local martingale deflator, in a local as well as in a global perspective (see Sections \ref{S3} and \ref{S4}, respectively). 
In that regard, our paper contributes to the literature dealing with utility maximization problems under partial information, see e.g. the papers \cite{BDL,L1,L2,Stet}.

Of course, our results are reminiscent of the dual approach to portfolio optimization, developed in the general semimartingale setting in the papers \cite{KS,S}, where optimal portfolios (solutions to the primal problem) are linked to supermartingale deflators (solutions to the dual problem). 
In the general setting of \cite{KS,S}, it is shown that the marginal utility of optimal terminal wealth defines a martingale measure with respect to a price system which uses as numéraire the optimal wealth process itself. In contrast, in the present paper we focus on characterizing the validity of the marginal utility measure for the original price system, without the need of changing the numéraire. Note also that, in the classical dual approach to portfolio optimization, the standing assumption is that the set of equivalent (local-/$\sigma$-) martingale measures is non-empty. Here, we opt for a different route and show that the existence of martingale measures / local martingale deflators and, hence, the no-arbitrage properties of the model, come as consequences of the viability of the financial market.

This paper is also closely related to the recent strand of literature that deals with financial models going beyond the traditional setting based on EMMs, relaxing the NFLVR requirement. One of the first studies in this direction is the paper \cite{LW}, where the authors are concerned with the viability of a complete It\^o-process model. In particular, they show that the financial market can be viable, in the sense that portfolio optimization problems can be meaningfully solved, even if the NFLVR condition does not necessarily hold (in a related context, see also \cite{FR} for an analysis of pricing and hedging problems in the absence of EMMs). In a general semimartingale framework, Proposition 4.19 of \cite{KK} shows that the minimal no-arbitrage requirement in order to solve expected utility maximization problems amounts to the \emph{no unbounded profit with bounded risk} (NUPBR) condition, the latter being weaker than NFLVR. 
Moreover, it has been recently proven in \cite{CDM} that the NUPBR condition is equivalent to the local solvability of portfolio optimization problems in a general semimartingale setting. The idea of adopting a local point of view is also taken up in the present paper, although in our context it is motivated by the need of ensuring good integrability properties, since we rely on BSDE methods for solving portfolio optimization problems. Moreover, we put a strong emphasis on the explicit connection between (local) market viability and the validity of the marginal utility measure.

The paper is structured as follows. Section \ref{S2} presents the modeling framework. In Section \ref{S3}, we prove the equivalence between local market viability and the existence of a PIEMM up to a stopping time. More specifically, this requires first to characterize optimal portfolios in terms of the solutions to an associated BSDE (Section \ref{S3.1}) and then to characterize the family of density processes of all PIEMMs (Section \ref{S3.2}); the main equivalence result is then proved in Section \ref{S3.3}. Section \ref{S4} deals with the issue of the global viability of the financial market, first in the simpler case of bounded coefficients (Section \ref{S4.1}) and then in the more general locally bounded case (Section \ref{S4.2}). Section \ref{S5} closes the paper by illustrating some of the main concepts and results in the context of a simple example. Finally, the Appendix collects some technical proofs of intermediate results.

\section{The modeling framework}	\label{S2}

On a given probability space $(\Omega,\G,P)$, let us consider a Brownian motion $B=\{B(t);t\geq 0\}$ and a homogeneous Poisson random measure $N(\cdot,\cdot)$ on $\R_+\!\times\R$, in the sense of Definition II.1.20 of \cite{JS}, independent of $B$. Let $\G=(\G_t)_{t\geq 0}$ be the filtration generated by $B$ and $N$, assumed to satisfy the usual conditions of right-continuity and $P$-completeness, and denote by $\cP_{\GG}$ the predictable $\sigma$-field of $\GG$. We denote by $m(dt,d\zeta):=\nu(d\zeta)\,dt$ the compensator of the random measure $N(dt,d\zeta)$, where $\nu$ is a $\sigma$-finite measure on $\bigl(\R,\B(\R)\bigr)$, and by $\tN(dt,d\zeta):=N(dt,d\zeta)-\nu(d\zeta)\,dt$ the corresponding compensated random measure. Finally, we let $T\in(0,\infty)$ represent a fixed investment horizon.

As mentioned in the introduction, we are interested in financial models where agents do not have access to the \emph{full information} filtration $\GG$. To this effect, we introduce a filtration $\FF=(\F_t)_{t\geq0}$, which represents the \emph{partial information} actually available. We assume that the filtration $\FF$ satisfies the usual conditions and that $\F_t\subseteq\G_t$, for all $t\geq0$, and we denote by $\cP_{\FF}$ the predictable $\sigma$-field of $\FF$.

We consider an abstract financial market with two investment possibilities (all the results of the present paper can be generalised to multi-dimensional models without significant difficulties):
\begin{myenumerate}
\item a risk-free asset with unit price $S^0(t) = 1$, for all $t\in[0,T]$;
\item a risky asset, with unit (discounted) price $S(t)$ given by the solution to the SDE
\end{myenumerate}
\be	\label{S}	\left\{ \ba
dS(t) & = b(t)\,dt + \sigma(t)\,dB(t) + \int_\R\!\gamma(t,\zeta)\,\tN(dt,d\zeta), \quad t\in[0,T]; \\
S(0) &= S_0\in\R,
\ea	\right. \ee
where the two processes $b=\bigl\{b(t);t\in[0,T]\bigr\}$ and $\sigma=\bigl\{\sigma(t);t\in[0,T]\bigr\}$ are $\cP_{\GG}$-measurable and $\gamma:\Omega\times[0,T]\times\R\rightarrow\R$ is a predictable function in the sense of \cite{JS}, i.e., $\cP_{\GG}\otimes\B(\R)$-measurable, and integrable with respect to $N$. We refer the reader to Section II.1 of \cite{JS} and to the monograph \cite{OS1} for more information on stochastic calculus with respect to Poisson random measures. We also impose the following assumption on $b$, $\sigma$ and $\gamma$ and on the sub-filtration $\FF$:

\begin{Ass}	\label{lcl-bdd}
\begin{myenumerate}
\item
The $\cP_{\GG}$-measurable processes $b=\bigl\{b(t);t\in[0,T]\bigr\}$ and $\sigma=\bigl\{\sigma(t);t\in[0,T]\bigr\}$ as well as the $\cP_{\GG}\otimes\B(\R)$-measurable function $\gamma$ are $\FF$-locally bounded;
\item
$\FF^S\subseteq\FF$, i.e., $\F_t^S\subseteq\F_t$ for all $t\in[0,T]$, where $\F_t^S$ is the $\sigma$-algebra generated by $\bigl\{S(u);u\in[0,t]\bigr\}$.
\end{myenumerate}
\end{Ass}

As can be easily verified, part (i) of Assumption \ref{lcl-bdd} implies that there exists a common sequence of $\FF$-stopping times $\{\tau_n\}_{n\in\N}$ with $\tau_n\nearrow+\infty$ $P$-a.s. as $n\rightarrow+\infty$ such that the stopped processes $S^{\tau_n}$, $b^{\tau_n}$, $\sigma^{\tau_n}$ and $\gamma(\cdot\wedge\tau_n,\cdot)$ are $P$-a.s. uniformly bounded, for every $n\in\N$.
In particular, part (i) of Assumption \ref{lcl-bdd} is always satisfied for the processes $b$ and $\sigma$ if they are $\cP_{\FF}$-measurable and left-continuous or right-continuous with limits from the left (see e.g. \cite{HWY}, Theorem 7.7). In view of part (ii) of the assumption, this is for instance the case if $b(t)$ and $\sigma(t)$ are given as continuous functions of $S(t-)$.
Part (ii) of Assumption \ref{lcl-bdd} implies that every market participant can observe the evolution of the (discounted) price of the risky asset $S$. Note, however, that the filtration $\FF^S$ is in general strictly smaller than $\GG$, since the observation of the price process $S$ does not suffice to unveil the two sources of randomness $B$ and $N$.

\begin{Rem}	\label{gen-case}
It is worth pointing out that the main results of the present paper can be obtained under weaker assumptions, at the expenses of greater technicalities. Indeed, Assumption \ref{lcl-bdd} can be significantly relaxed, by only assuming that $b$, $\sigma$ and $\gamma$ are $\GG$-locally bounded, i.e., with respect to a sequence $\{\tilde{\tau}_n\}_{n\in\N}$ of $\GG$-stopping times, without any further assumption on the sub-filtration $\FF$ \footnote{This would then allow to also consider the cases of \emph{delayed information}, where $\F=\G_{(t-\delta)^+}$, for some $\delta>0$, and of \emph{discrete observations}, where $\F_t=\bar{\F}^S_t$, with $\bar{\F}^S_t$ being the $\sigma$-algebra generated by $\{S(t_i);0=t_0<t_1<\ldots<t_n\leq t\}$, $n\in\N$.}. In that case, one can formulate all our results in terms of optional projections on the sub-filtration $\FF$ (see e.g. \cite{HWY}, Chapter V)\footnote{Let us recall that, for any bounded measurable process $Y=\bigl\{Y(t);t\in[0,T]\bigr\}$, the $\FF$-optional projection is defined as the unique $\FF$-optional bounded process $^oY=\bigl\{^oY(t);t\in[0,T]\bigr\}$ such that $E\bigl[Y(\tau)\ind_{\{\tau<\infty\}}|\F_{\tau}\bigr]=\,^oY(\tau)\ind_{\{\tau<\infty\}}$ $P$-a.s. for every $\FF$-stopping time $\tau$. In particular, we have $^oY(t)=E\bigl[Y(t)|\F_t\bigr]$ $P$-a.s. for every $t\in[0,T]$.}. However, since this would lead to less transparent and more technical definitions of the main concepts, we prefer to impose the stronger Assumption \ref{lcl-bdd}.
\end{Rem}

We say that a function $U:(-\infty,\infty]\rightarrow[-\infty,\infty)$ of class $\mathcal{C}^1$ on $(-\infty,\infty)$ is a \emph{utility function} if it is concave and strictly increasing on $(-\infty,\infty)$ and we denote by $U'$ its first derivative (\emph{marginal utility}).
Aiming at describing the activity of trading in the financial market on the basis of the partial information represented by the sub-filtration $\FF$ and according to the preference structure represented by the utility function $U$, we define the family $\A_{\FF}^U$ of \emph{admissible strategies} as follows:
\be	\label{adm}
\A^U_{\FF} := \left\{\text{all $\cP_{\FF}$-measurable processes }\phi=\bigl\{\phi(t);t\in[0,T]\bigr\}
\text{ s.t. }X_{\phi}\in\calS^2
\text{ and }E\bigl[U'\bigl(X_{\phi}(T)\bigr)^2\bigr]<\infty\right\}
\ee
where $\phi_t$ represents the number of units of the risky asset held in the portfolio at time $t$, for all $t\in[0,T]$, with associated wealth process $X_{\phi}=\bigl\{X_{\phi}(t);t\in[0,T]\bigr\}$, and where $\calS^2$ denotes the family of all $\GG$-semimartingales $Y=\bigl\{Y(t);t\in[0,T]\bigr\}$ satisfying $E\bigl[\,\sup_{t\in[0,T]}\left|Y(t)\right|^2\bigr]<\infty$. The requirement of $\FF$-predictability amounts to ensuring that agents trade by relying only on the partial information at their disposal, while the square-integrability requirement is an indispensable integrability condition in order to solve utility maximization problems via BSDE techniques, as shown in the next section.

As usual, we assume that trading is done in a self-financing way, so that the wealth process associated to a given strategy $\phi\in\A_{\FF}^U$ starting from an initial endowment $x\in\R$ is given by
\be \label{wealth}	\left\{	\ba
dX_\varphi(t) & = \varphi(t)\,dS(t) 
= \varphi(t) \left(b(t)\,dt + \sigma(t)\,dB(t) + \int_\R\!\gamma(t,\zeta)\,\tN(dt,d\zeta)\right), 
\quad t\in[0,T]; \\
X_\varphi(0) &= x.
\ea	\right. \ee

\begin{Rem}
\begin{myenumerate}
\item
Due to part (ii) of Assumption \ref{lcl-bdd} together with Theorem 3.1 of \cite{Str}, the process $S$ is an $\FF$-semimartingale. In particular, this implies that the wealth process $X_{\phi}$ is well-defined as a stochastic integral and, hence, a semimartingale in the partial information filtration $\FF$.
\item
Defining the class of admissible strategies as in \eqref{adm} ensures that $E\bigl[U\bigl(X_{\phi}(T)\bigr)\bigr]<\infty$ for all $\phi\in\A_{\FF}^U$. Indeed, due to the concavity of $U$, we have $U\bigl(X_{\phi}(T)\bigr)\leq U(x)+U'(x)\bigl(X_{\phi}(T)-x\bigr)$. Since $X_{\phi}\in\calS^2$, this implies that $U\bigl(X_{\phi}(T)\bigr)\in L^2(P)$.
\end{myenumerate}
\end{Rem}

We want to emphasize that we do not a-priori impose any no-arbitrage restriction on the financial market model. In the remaining part of the paper, the no-arbitrage properties of the model will be inferred as consequences of the (local) solvability of a portfolio optimization problem.

\section{Local market viability under partial information}	\label{S3}

In the present section we prove the equivalence between the concept of local market viability, introduced below in Definition \ref{viability-loc}, and the local existence of a partial information equivalent martingale measure (PIEMM; see Definition \ref{PIEMM-loc}) such that its density is expressed in terms of the marginal utility of terminal wealth. 
To this effect, as a first step (Section \ref{S3.1}), we shall characterize the solutions to portfolio optimization problems by applying suitable versions of the maximum principles developed in \cite{BO} for stochastic control problems under partial information. As a second step (Section \ref{S3.2}), we provide a characterization of the density processes of PIEMMs. 
From a technical point of view, the need to embark on a local analysis arises from the integrability properties required in the above two steps (see in particular the proofs of Propositions \ref{BSDE-opt} and \ref{PIEMM}). Hence, in order to have  good integrability properties, we shall rely on the local boundedness assumption (see part (i) of Assumption \ref{lcl-bdd}; under stronger assumptions, a direct global result will be proved in Section \ref{S4.1}).
Until the end of Section \ref{S3}, we fix an element $n\in\N$ and let $\tau_n$ be the corresponding $\FF$-stopping time from the sequence $\{\tau_n\}_{n\in\N}$ introduced after Assumption \ref{lcl-bdd}. We denote by $b^n$ the stopped process $b^n:=\bigl\{b(t\wedge\tau_n);t\in[0,T]\bigr\}$, with an analogous notation for $\sigma^n$ and $\gamma^n$.

\begin{Pb}[Partial information locally optimal portfolio problem]	\label{utility-loc}
For a fixed $n\in\N$, for a given utility function $U$ and an initial endowment $x\in\R$, find an element $\phi^{*,n}\in\A_{\FF}^U(n)$ such that 
\[
\underset{\phi\in\A_{\FF}^U(n)}{\sup}E\Bigl[U\bigl(X_{\phi}(T\wedge\tau_n)\bigr)\Bigr]
= E\Bigl[U\bigl(X_{\phi^{*,n}}(T\wedge\tau_n)\bigr)\Bigr]
<\infty
\]
where $\A^U_{\FF}(n):=\Bigl\{\text{all $\cP_{\FF}$-measurable processes }\phi=\bigl\{\phi(t);t\in[0,T]\bigr\}\text{ s.t. }\phi\ind_{\lsi0,\tau_n\rsi}\in\A^U_{\FF}\Bigr\}$.
\end{Pb}

We now define the notion of market viability in terms of the solvability of the partial information locally optimal portfolio problem\footnote{We want to mention that defining the notion of market viability in terms of the solvability of portfolio optimization problems has a long history in financial mathematics, going back to the seminal paper \cite{HK} (see also \cite{LW} in a related context).}.

\begin{Def}[Local market viability]	\label{viability-loc}
Let $U$ be a utility function. The financial market is said to be \emph{locally viable up to $\tau_n$} if Problem \ref{utility-loc} admits an optimal solution $\phi^{*,n}\in\A_{\FF}^U(n)$.
\end{Def}

\subsection{A BSDE characterization of locally optimal portfolios}	\label{S3.1}

As a first step, we provide a characterization of the locally optimal portfolio which solves Problem \ref{utility-loc} in terms of the solution to a backward stochastic differential equation (BSDE), by relying on necessary and sufficient maximum principles for stochastic control under partial information, see e.g. \cite{BO,OS2}. 

We define the Hamiltonian $H^n:\Omega\times[0,T]\times\R^3\times\mathcal{R}\rightarrow\R$ as follows:
\be	\label{Hamil}
H^n\bigl(\omega,t,\phi,p,q,r(\cdot)\bigr) 
:= \phi\,b^n(\omega,t)\,p
+\phi\,\sigma^n(\omega,t)\,q
+\phi\int_\R\!r(\zeta)\,\gamma^n(\omega,t,\zeta)\nu(d\zeta)
\ee
where $\mathcal{R}$ is defined as the class of functions $r:\R\setminus\{0\}\rightarrow\R$ such that the integral in \eqref{Hamil} converges. To the Hamiltonian $H^n$ we associate a BSDE for the adjoint processes $p^n=\bigl\{p^n(t);t\in[0,T]\bigr\}$, $q^n=\bigl\{q^n(t);t\in[0,T]\bigr\}$ and for the function $r^n:\Omega\times[0,T]\times\R\rightarrow\R$ as follows, for any $\phi\in\A_{\FF}^U(n)$:
\be	\label{BSDE}	\left\{	\ba
dp^n(t) &= q^n(t)\,dB(t)+\int_\R\!r^n(t,\zeta)\,\tN(dt,d\zeta),	\quad t\in[0,T];\\
p^n(T) &= U'\bigl(X_{\phi}(T\wedge\tau_n)\bigr)\,.
\ea	\right.	\ee
In order to study the BSDE \eqref{BSDE}, we need to introduce the following classes of processes:
\begin{gather*}
\M^2 := \left\{\text{all $\GG$-martingales }M=\bigl\{M(t);t\in[0,T]\bigr\}
\text{ s.t. }E\biggl[\,\sup_{t\in[0,T]}|M(t)|^2\biggr]<\infty\right\};	\\
L^2(B) :=  \left\{\text{all $\cP_{\GG}$-measurable processes }q=\bigl\{q(t);t\in[0,T]\bigr\}
\text{ s.t. }E\biggl[\int_0^T\!\!q^2(t)dt\biggr]<\infty\right\};	\\
G^2(\tN) :=  \left\{\text{all $\cP_{\GG}\otimes\B(\R)$-measurable functions }r
\text{ s.t. }E\biggl[\int_0^T\!\!\int_\R r^2(t,\zeta)\nu(d\zeta)dt\biggr]<\infty\right\}.
\end{gather*}

\begin{Lem}	\label{BSDE-sol}
For any fixed $n\in\N$ and $\phi\in\A^U_{\FF}(n)$, the BSDE \eqref{BSDE} admits a unique solution $(p^n,q^n,r^n)\in\M^2\times L^2(B)\times G^2(\tN)$. Furthermore, the $\GG$-martingale $p^n$ satisfies $E\left[\int_0^T\!\bigl(p^n(t)\bigr)^2dt\right] < \infty$.
\end{Lem}
\begin{proof}
See the Appendix.
\end{proof}

\begin{Prop}	\label{BSDE-opt}
For any fixed $n\in\N$, an element $\phi\in\A_{\FF}^U(n)$ solves Problem \ref{utility-loc} if and only if the solution $(p^n,q^n,r^n)\in\M^2\times L^2(B)\times G^2(\tN)$ to the corresponding BSDE \eqref{BSDE} satisfies the following condition $P$-a.s. for a.a. $t\in[0,T\wedge\tau_n]$:
\be	\label{BSDE-cond}
E\left[\frac{\partial H}{\partial \phi}\bigl(t,\phi(t),p^n(t),q^n(t),r^n(t)\bigr)\Bigr|\F_t\right]
= E\biggl[b^n(t)\,p^n(t) + \sigma^n(t)\,q^n(t)
+\int_\R\!\gamma^n(t,\zeta)\,r^n(t,\zeta)\,\nu(d\zeta)\Bigr|\F_t\biggr]=0\,.
\ee
\end{Prop}
\begin{proof}
See the Appendix.
\end{proof}

Condition \eqref{BSDE-cond} also admits an alternative formulation, in terms of the generalised Malliavin derivatives of the marginal utility $U'$. To this effect, recall the generalised Clark-Ocone theorem (see \cite{Aa} for the Brownian motion case and Theorem 3.28 of \cite{DOP} for the L\'evy process case) which states that if the random variable $F \in L^2(P)$ is $\G_T$-measurable, then it can be written as
\[
F = E[F] + \int_0^T\!\!E[D_t F\,|\G_t]\,dB(t) + \int_0^T\!\!\int_\R E[D_{t,\zeta} F\,|\G_t]\,\tN(dt,d\zeta)
\]
where $D_t$ and $D_{t,\zeta}$ denote the generalised Malliavin derivatives at $t$ with respect to $B$ and at $t,\zeta$ with respect to $N$, respectively. Applying this to $F := U'\bigl(X_{\phi}(T\wedge\tau_n)\bigr)$ we see that the solution $(p^n,q^n,r^n)$ to the BSDE \eqref{BSDE} can be represented as follows, for all $t\in[0,T]$ and $\zeta\in\R$:
\begin{align*}
p^n(t) &=  E\bigl[U'\bigl(X_\varphi(T\wedge\tau_n)\bigr) | \G_t\bigr], \\
q^n(t)  &= E\bigl[D_t U'\bigl(X_\varphi(T\wedge\tau_n)\bigr) | \G_t\bigr],	\\
r^n(t,\zeta) &= E\bigl[D_{t,\zeta} U'\bigl(X_\varphi(T\wedge\tau_n)\bigr) | \G_t\bigr].
\end{align*}

Therefore, in view of Proposition \ref{BSDE-opt}, we get the following characterization of the optimal terminal wealth $X_{\phi^{*,n}}(T\wedge\tau_n)$ of the partial information locally optimal portfolio problem (Problem \ref{utility-loc}).

\begin{Cor}
For any fixed $n\in\N$, an element $\phi\in\A_{\FF}^U(n)$ solves Problem \ref{utility-loc} if and only if the corresponding terminal wealth $X_{\phi}(T\wedge\tau_n)$ satisfies the following partial information Malliavin differential equation $P$-a.s. for a.a. $t\in[0,T\wedge\tau_n]$:
\[
E\biggl[b^n(t)\,U'\bigl(X_{\phi}(T\wedge\tau_n)\bigr)
+\sigma^n(t)\,D_tU'\bigl(X_{\phi}(T\wedge\tau_n)\bigr)
+\int_\R\!\gamma^n(t,\zeta)\,D_{t,\zeta}U'\bigl(X_{\phi}(T\wedge\tau_n)\bigr)\nu(d\zeta)\Bigr|\F_t\biggr] = 0.
\]
\end{Cor}

\subsection{Partial information equivalent martingale measures (PIEMMs)}	\label{S3.2}

We now proceed to characterize the density processes of all \emph{partial information equivalent martingale measures} (PIEMMs), defined below in Definition \ref{PIEMM-loc}. As a preliminary, let us consider a generic probability measure $Q\sim P$ on $(\Omega,\G_T)$ and denote by $G=\bigl\{G(t);t\in[0,T]\bigr\}$ its density process, i.e., $G(t):=\frac{dQ|_{\G_t}}{dP|_{\G_t}}$ for all $t\in[0,T]$. It is well-known that $G$ is a strictly positive $\GG$-martingale with $E[G(T)]=1$. 
Observe that the density process $G$ has been defined with respect to the full information filtration $\GG$, not to the partial information filtration $\FF$. The reason is that, in the current general setting, we do not assume any specific structure for $\FF$, and, hence, it would not be possible to provide any useful characterization of the density processes in $\FF$ \footnote{We want to point out that, for the same reason, the BSDE \eqref{BSDE} has been formulated with respect to the full information filtration $\GG$, even though the terminal condition $U'\bigl(X_{\phi}(T\wedge\tau_n)\bigr)$ is clearly $\F_T$-measurable.}. 
On the contrary, due to the martingale representation property in the filtration $\GG$, there exists a $\cP_{\GG}$-measurable process $\theta_0=\bigl\{\theta_0(t);t\in[0,T]\bigr\}$ with $\int_0^T\!\theta_0^2(t)dt<\infty$ $P$-a.s. and a $\cP_{\GG}\otimes\B(\R)$-measurable function $\theta_1:\Omega\times[0,T]\times\R\rightarrow(-1,\infty)$ with $\int_0^T\!\int_\R \theta_1^2(t,\zeta)\nu(d\zeta)dt<\infty$ $P$-a.s. such that the following holds:
\be	\label{density}	\left\{	\ba
dG(t) &= G(t-)\left(\theta_0(t)\,dB(t)+\int_\R\!\theta_1(t,\zeta)\,\tN(dt,d\zeta)\right), 
\quad t\in[0,T];	\\
G(0) &= 1.
\ea	\right.	\ee
For all $t\in[0,T]$, the SDE \eqref{density} admits as explicit solution
\be	\label{density-2}	\ba
G(t) &= \exp\left(\int_0^t\!\theta_0(s)\,dB(s)-\frac{1}{2}\int_0^t\!\theta_0^2(s)\,ds
+\int_0^t\!\!\int_\R\!\log\bigl(1+\theta_1(s,\zeta)\bigr)\,\tN(ds,d\zeta)\right.	\\
&\phantom{=\exp=}\left.
+\int_0^t\!\!\int_\R\!\bigl\{\log\bigl(1+\theta_1(s,\zeta)\bigr)-\theta_1(s,\zeta)\bigr\}\nu(d\zeta)\,ds\right).
\ea	\ee
In the following, we write $G_{\theta}(t):=G(t)$, for $\theta:=(\theta_0,\theta_1)$, where $G(t)$ is represented by $\theta$ as above. We let $\Theta$ denote the family of all $\cP_{\GG}$-measurable processes $\theta=(\theta_0,\theta_1)$ such that the SDE \eqref{density} has a unique strictly positive martingale solution $G_{\theta}=\bigl\{G_{\theta}(t);t\in[0,T]\bigr\}$. Similarly, for $\theta\in\Theta$, we denote by $Q_{\theta}$ the measure on $(\Omega,\G_T)$ defined by $dQ_{\theta}/dP:=G_{\theta}(T)$ and by $E_{Q_\theta}[\cdot]$ the corresponding expectation.

\begin{Def}	\label{PIEMM-loc}
For a fixed $n\in\N$, a probability measure $Q_{\theta}\sim P$ on $(\Omega,\G_T)$ is said to be a \emph{partial information equivalent martingale measure (PIEMM) up to $\tau_n$} if the process $S^{\tau_n}$ is a $(Q_{\theta},\FF)$-martingale\footnote{If $\tau_n$ is a $\GG$-stopping time, but not necessarily an $\FF$-stopping time, and $\FF^S\subseteq\FF$ does not necessarily hold, a PIEMM up to $\tau_n$ can be defined in terms of the $(Q_{\theta},\FF)$-martingale property of the $(Q_{\theta},\FF)$-optional projection of the stopped process $S^{\tau_n}$ (compare with Remark \ref{gen-case}). In that case, analogous versions of Proposition \ref{PIEMM} and Theorem \ref{viability-PIEMM-loc} can be established.}.
\end{Def}

In the next proposition we characterize the density processes of all PIEMMs.

\begin{Prop}	\label{PIEMM}
For any fixed $n\in\N$, a probability measure $Q_{\theta}\sim P$ on $(\Omega,\G_T)$ such that $(\theta_0,\theta_1)\in L^2(B)\times G^2(\tN)$ is a PIEMM up to $\tau_n$ if and only if the following condition holds:
\be	\label{PIEMM-cond}
E_{Q_{\theta}}\!\left[b^n(t)+\sigma^n(t)\,\theta_0(t\wedge\tau_n)
+\int_\R\!\gamma^n(t,\zeta)\,\theta_1(t\wedge\tau_n,\zeta)\nu(d\zeta)\Bigr|\F_t\right] = 0
\qquad\text{$P$-a.s. for a.a. }t\in[0,T\wedge\tau_n].
\ee
\end{Prop}
\begin{proof}
See the Appendix.
\end{proof}

\subsection{Local market viability and PIEMMs}	\label{S3.3}

We now combine the results of Sections \ref{S3.1}-\ref{S3.2} to obtain our first main result, namely a characterization of local market viability in terms of partial information equivalent martingale measures.

\begin{Thm}	\label{viability-PIEMM-loc}
For any fixed $n\in\N$, the following are equivalent:
\begin{myenumerate}
\item 
the financial market is locally viable up to $\tau_n$ (in the sense of Definition \ref{viability-loc}) and $\phi^*\in\A_{\FF}^U(n)$ solves the partial information locally optimal portfolio problem (Problem \ref{utility-loc});
\item
for $\phi^*\in\A_{\FF}^U(n)$, the measure $Q^{\phi^*,n}\sim P$ defined on $(\Omega,\G_T)$ by
\be	\label{thm-1}
\frac{dQ^{\phi^*,n}}{dP} := \frac{U'\bigl(X_{\phi^*}(T\wedge\tau_n)\bigr)}
{E\bigl[U'\bigl(X_{\phi^*}(T\wedge\tau_n)\bigr)\bigr]}
\ee
is a PIEMM up to $\tau_n$, in the sense of Definition \ref{PIEMM-loc}.
\end{myenumerate}
\end{Thm}
\begin{proof}
{\bf (i) $\Rightarrow$ (ii)}:
Due to Proposition \ref{BSDE-opt}, if $\phi\in\A_{\FF}^U(n)$ solves Problem \ref{utility-loc}, then the unique solution $(p^n,q^n,r^n)\in\M^2\times L^2(B)\times G^2(\tN)$ to the BSDE \eqref{BSDE} satisfies condition \eqref{BSDE-cond}. Let us define, for all $t\in[0,T]$ and $\zeta\in\R$
\be	\label{def-proc}
G(t) := \frac{p^n(t)}{p^n(0)} 
= \frac{E\bigl[U'\bigl(X_{\phi}(T\wedge\tau_n)\bigr)|\G_t\bigr]}
{E\bigl[U'\bigl(X_{\phi}(T\wedge\tau_n)\bigr)\bigr]},
\qquad
\theta_0(t) := \frac{q^n(t)}{p^n(t)},
\qquad
\theta_1(t,\zeta) := \frac{r^n(t,\zeta)}{p^n(t-)}.
\ee
Then, by combining equation \eqref{BSDE} with \eqref{def-proc}, we get
\[	\ba
dG(t) &= \frac{dp^n(t)}{p^n(0)} 
= \frac{q^n(t)}{p^n(0)}\,dB(t)+\int_\R\!\frac{r^n(t,\zeta)}{p^n(0)}\,\tN(dt,d\zeta) 
= \frac{p^n(t)}{p^n(0)}\,\theta_0(t)\,dB(t)+\frac{p^n(t-)}{p^n(0)}\int_\R\!\theta_1(t,\zeta)\tN(dt,d\zeta) \\
&= G(t-)\left[\theta_0(t)\,dB(t)+\int_\R\!\theta_1(t,\zeta)\,\tN(dt,d\zeta)\right].
\ea	\]
By letting $dQ^{\phi,n}/dP:=G(T)$, we get a well-defined probability measure $Q^{\phi,n}\sim P$ on $(\Omega,\G_T)$ with density given by the right-hand side of \eqref{thm-1}. In view of Proposition \ref{PIEMM}, in order to show that $Q^{\phi,n}$ is a PIEMM up to $\tau_n$ it suffices to show that condition \eqref{PIEMM-cond} holds. This follows immediately by substituting \eqref{def-proc} into condition \eqref{BSDE-cond}.\\
\noindent {\bf (ii) $\Rightarrow$ (i)}:
Suppose that the probability measure defined by the right-hand side of \eqref{thm-1} is a PIEMM up to $\tau_n$, for some $\phi\in\A_{\FF}^U(n)$, and define the process $G=\bigl\{G(t);t\in[0,T]\bigr\}$ as follows:
\[
G(t) := \frac{E\bigl[U'\bigl(X_{\phi}(T\wedge\tau_n)\bigr)|\G_t\bigr]}
{E\bigl[U'\bigl(X_{\phi}(T\wedge\tau_n)\bigr)\bigr]},
\qquad \text{ for all }t\in[0,T].
\]
By the martingale representation property, the process $G$ admits a representation of the form \eqref{density}, for some $\cP_{\GG}$-measurable process $\theta_0$ and for some $\cP_{\GG}\otimes\B(\R)$-measurable function $\theta_1$. Furthermore, since $Q^{\phi,n}$ is a PIEMM up to $\tau_n$, Proposition \ref{PIEMM} implies that the following condition is satisfied $P$-a.s. for a.a. $t\in[0,T\wedge\tau_n]$:
\be	\label{thm-3}
E\left[G(t)\left(b^n(t)+\sigma^n(t)\,\theta_0(t\wedge\tau_n)
+\int_\R\!\gamma^n(t,\zeta)\,\theta_1(t\wedge\tau_n,\zeta)\nu(d\zeta)\right)\Bigr|\F_t\right] = 0\,.
\ee
Let us then define, for all $t\in[0,T]$ and $\zeta\in\R$
\be	\label{thm-4}
p^n(t) := E\bigl[U'\bigl(X_{\phi}(T\wedge\tau_n)\bigr)\bigr]G(t),
\qquad
q^n(t) := p^n(t)\,\theta_0(t\wedge\tau_n),
\qquad
r^n(t,\zeta) := p^n(t-)\theta_1(t\wedge\tau_n,\zeta).
\ee
Note that, since $U'\bigl(X_{\phi}(T\wedge\tau_n)\bigr)\in L^2(P)$, we have $(p^n,q^n,r^n)\in\M^2\times L^2(B)\times G^2(\tN)$. By substituting \eqref{thm-4} into \eqref{density}, we see that $(p^n,q^n,r^n)$ satisfies the BSDE \eqref{BSDE}. Moreover, by substituting \eqref{thm-4} into \eqref{thm-3}, we can verify that \eqref{BSDE-cond} holds. Proposition \ref{BSDE-opt} allows then to conclude that $\phi\in\A_{\FF}^U(n)$ solves Problem \ref{utility-loc}.
\end{proof}

\begin{Rem}	\label{duality}
As we have already mentioned, we did not a-priori introduce any no-arbitrage restriction on the model. The result of Theorem \ref{viability-PIEMM-loc} can then be interpreted in the following sense: as soon as the financial market is locally viable, in the sense that a portfolio optimization problem admits locally a solution, then there locally exists a partial information equivalent martingale measure. This means that the absence of arbitrage opportunities comes as a direct consequence of local market viability.
\end{Rem}

\section{Global market viability under partial information}	\label{S4}

So far, we have studied the viability of the financial market in a local sense, namely up to the stopping times composing the localizing sequence $\{\tau_n\}_{n\in\N}$. Now, we adopt a global perspective and aim at characterizing the global viability of the financial market on the whole investment horizon $[0,T]$. In Section \ref{S4.1}, we directly prove a global version of Theorem \ref{viability-PIEMM-loc} under the stronger assumption that $b$, $\sigma$ and $\gamma$ in \eqref{S} are bounded (and not only locally bounded) and the measure $\nu$ is finite. In Section \ref{S4.2}, we shall deal with the more delicate locally bounded case. 

\subsection{The case of bounded coefficients}	\label{S4.1}

This subsection aims at proving, under a rather strong assumption (see Assumption \ref{ass-easy} below), the equivalence between the existence of a PIEMM and the viability of the financial market in a global sense. In the spirit of Definition \ref{viability-loc}, we adopt the following definition of global market viability.

\begin{Def}	\label{viability-glob-easy}
Let $U$ be a utility function. The financial market is said to be \emph{globally viable} if there exists an element $\phi^*\in\A^U_{\FF}$ such that
\be	\label{utility-max}
\underset{\phi\in\A^U_{\FF}}{\sup}E\Bigl[U\bigl(X_{\phi}(T)\bigr)\Bigr]
= E\Bigl[U\bigl(X_{\phi^*}(T)\bigr)\Bigr]<\infty.
\ee
\end{Def}

In general, it turns out that the equivalence between the global viability of the financial market (in the sense of Definition \ref{viability-glob-easy}) and the existence of a PIEMM with density given by the (normalised) marginal utility of the optimal terminal wealth does not hold, as shown by an explicit counterexample in Section \ref{S5}. However, we can still obtain a direct and global version of Theorem \ref{viability-PIEMM-loc} if the following stronger assumption is satisfied.

\begin{Ass}	\label{ass-easy}
\begin{myenumerate}
\item
The $\cP_{\GG}$-measurable processes $b$, $\sigma$ as well as the $\cP_{\GG}\otimes\B(\R)$-measurable function $\gamma$ are $P$-a.s. uniformly bounded and the measure $\nu$ is finite;
\item
the utility function $U$ satisfies the following condition:
\be	\label{utility-cond}
E\Bigl[U'\bigl(X_{\phi}(T)+\xi\bigr)^2\Bigr]<\infty,
\qquad\text{for all }\phi\in\A^U_{\FF}\text{ and for all }\xi\in\!\!\bigcap_{r\in(1,\infty)}\!\!\!L^r(P).
\ee
\end{myenumerate}
\end{Ass}

A useful consequence of part (i) of Assumption \ref{ass-easy} is that the price process $S$ admits finite moments of every order, as shown in the next simple lemma, the proof of which is postponed to the Appendix.

\begin{Lem}	\label{moments}
If Assumption \ref{ass-easy}-(i) holds, then $E\bigl[\,\sup_{t\in[0,T]}|S(t)|^r\bigr]<\infty$ for all $r\in(1,\infty)$.
\end{Lem}

Following the same approach of Section \ref{S3.1}, we can characterize the solution to the portfolio optimization problem \eqref{utility-max} via the solution $(p,q,r)\in\M^2\times L^2(B)\times G^2(\tN)$ to the associated BSDE
\be	\label{BSDE-glob}	\left\{	\ba
dp(t) &= q(t)\,dB(t)+\int_\R\!r(t,\zeta)\,\tN(dt,d\zeta),	\quad t\in[0,T];\\
p(T) &= U'\bigl(X_{\phi}(T)\bigr).
\ea	\right.	\ee

\begin{Prop}	\label{BSDE-opt-easy}
Suppose that Assumption \ref{ass-easy} holds. An element $\phi\in\A^U_{\FF}$ solves problem \eqref{utility-max} if and only if the solution $(p,q,r)\in\M^2\times L^2(B)\times G^2(\tN)$ to the corresponding BSDE \eqref{BSDE-glob} satisfies the following condition $P$-a.s. for a.a. $t\in[0,T]$:
\[
E\biggl[b(t)\,p(t) + \sigma(t)\,q(t)
+\int_\R\!\gamma(t,\zeta)\,r(t,\zeta)\,\nu(d\zeta)\Bigr|\F_t\biggr]=0
\]
\end{Prop}
\begin{proof}
Due to Assumption \ref{ass-easy} together with Lemma \ref{moments}, the strategy $\bigl\{\beta(t);t\in[0,T]\bigr\}$ defined by $\beta(t):=\xi\ind_{[t_0,t_0+h]}(t)$ belongs to $\A_{\FF}^U$, for any $t_0\in[0,T]$, $h>0$, and for every bounded $\F_{t_0}$-measurable random variable $\xi$. Similarly, for every $\psi,\eta\in\A^U_{\FF}$ with $\eta$ bounded, we have $\psi+\delta\,\eta\in\A_{\FF}^U$ for any $\delta\in\R$. This shows that conditions (A1)-(A2) of the necessary maximum principle of \cite{BO} are satisfied. By relying on the boundedness of $b$, $\sigma$ and $\gamma$ as well as on Lemma \ref{BSDE-sol}, the same arguments used in the proof of Proposition \ref{BSDE-opt} allow then to prove the claim.
\end{proof}

Recall that, as in Section \ref{S3.2}, the density process $G=\bigl\{G(t);t\in[0,T]\bigr\}$ of a probability measure $Q\sim P$ on $(\Omega,\G_T)$ admits a representation of the form \eqref{density-2}, for some $\cP_{\GG}$-measurable process $\theta$ and for some $\cP_{\GG}\otimes\B(\R)$-measurable function $\theta_1:\Omega\times[0,T]\times\R\rightarrow(-1,\infty)$. Let us introduce the following definition, which is a natural extension of Definition \ref{PIEMM-loc} to the global investment horizon $[0,T]$.

\begin{Def}	\label{PIEMM-easy}
A probability measure $Q_{\theta}\sim P$ on $(\Omega,\G_T)$ is said to be a \emph{partial information equivalent martingale measure (PIEMM)} if the process $S$ is a $(Q_{\theta},\FF)$-martingale.
\end{Def}

Density processes $G_{\theta}$ of PIEMMs can be characterized as follows, similarly to Proposition \ref{PIEMM}.

\begin{Prop}	\label{PIEMM-dens}
Suppose that Assumption \ref{ass-easy}-(i) holds. A probability measure $Q_{\theta}\sim P$ on $(\Omega,\G_T)$ with $dQ_{\theta}/dP\in L^2(P)$ is a PIEMM if and only if the following condition holds:
\[
E_{Q_{\theta}}\!\left[b(t)+\sigma(t)\,\theta_0(t)
+\int_\R\!\gamma(t,\zeta)\,\theta_1(t,\zeta)\nu(d\zeta)\Bigr|\F_t\right] = 0
\qquad\text{$P$-a.s. for a.a. }t\in[0,T]
\]
where $\theta_0$ ($\theta_1$, resp.) is the process (predictable function, resp.) appearing in the representation \eqref{density-2}.
\end{Prop}
\begin{proof}
The claim can be proved by relying on the same arguments used in the proof of Proposition \ref{PIEMM} (of course, without stopping by $\tau_n$). Note that, in the context of this subsection, the true $\GG$-martingale property of the $\GG$-local martingale part in \eqref{proof-2} can be verified by relying on Lemma \ref{moments} together with 
standard inequalities, using the fact that $dQ_{\theta}/dP\in L^2(P)$.
\end{proof}

As in Section \ref{S3.3}, we can now combine Propositions \ref{BSDE-opt-easy} and \ref{PIEMM-dens} in order to obtain the equivalence between the global viability of the financial market (in the sense of Definition \ref{viability-glob-easy}) and the existence of a PIEMM. We omit the proof, which is entirely similar to the proof of Theorem \ref{viability-PIEMM-loc}.

\begin{Thm}	\label{viability-bdd}
Suppose that Assumption \ref{ass-easy} holds. Then the following are equivalent:
\begin{myenumerate}
\item
the financial market is globally viable, in the sense of Definition \ref{viability-glob-easy}, and $\phi^*\in\A^U_{\FF}$ solves the partial information optimal portfolio problem \eqref{utility-max};
\item
for $\phi^*\in\A^U_{\FF}$, the measure $Q^{\phi^*}\sim P$ on $(\Omega,\G_T)$ defined by
\[
\frac{dQ^{\phi^*}}{dP} := \frac{U'\bigl(X_{\phi^*}(T)\bigr)}{E\bigl[U'\bigl(X_{\phi^*}(T)\bigr)\bigr]}
\]
is a PIEMM, in the sense of Definition \ref{PIEMM}.
\end{myenumerate}
\end{Thm}


\subsection{The general case}	\label{S4.2}

In the present section, we study the issue of global market viability in the more general case where $b$, $\sigma$ and $\gamma$ are only $\FF$-locally bounded, as in Section \ref{S2}, without assuming that the stronger simplifying Assumption \ref{ass-easy} holds. In this case, as will be shown by the explicit example contained in Section \ref{S5}, we cannot characterize global market viability in terms of PIEMMs and we need to rely on the localization approach described in Section \ref{S3}, adopting the following definition of global market viability.

\begin{Def}[Global market viability]	\label{viability-glob}
Let $U$ be a utility function. The financial market is said to be \emph{globally viable} if, for all $n\in\N$, Problem \ref{utility-loc} admits an optimal solution $\phi^*_n\in\A_{\FF}^U(n)$.
\end{Def}

In general, one cannot hope to obtain a full characterization of global market viability in terms of (partial information) equivalent martingale measures defined on the whole investment horizon $[0,T]$ (not even in terms of partial information equivalent local martingale measures), as will be shown in Section \ref{S5}. Hence, we need to introduce the following notion\footnote{If the sequence $\{\tau_n\}_{n\in\N}$ is composed of $\GG$-stopping times, not necessarily of $\FF$-stopping times, and if the inclusion $\FF^S\subseteq\FF$ does not hold, then Definition \ref{deflator} can be weakened by requiring that the $\FF$-optional projection of $ZS$ is an $\FF$-local martingale. In that case, analogous versions of Theorems \ref{ex-defl} and \ref{viability-global} can be established.}, which corresponds to a weaker counterpart of the density process of a PIEMM and extends to the partial information setting the concept of local martingale deflator introduced in \cite{Kar}.

\begin{Def}	\label{deflator}
A strictly positive $\FF$-local martingale $Z=\bigl\{Z(t);t\in[0,T]\bigr\}$ with $Z(0)=1$ is said to be a \emph{partial information local martingale deflator (PILMD)} if $ZS$ is an $\FF$-local martingale.
\end{Def}

As shown in \cite{Kar}, the existence of a local martingale deflator is equivalent to the \emph{no arbitrage of the first kind} (NA1) condition, which is in turn equivalent to the absence of \emph{unbounded profits with bounded risk} (NUPBR), see \cite{KK}. In particular, the NA1/NUPBR condition can be shown to be strictly weaker than the classical \emph{no free lunch with vanishing risk} (NFLVR) condition introduced in \cite{DS0}, the latter being equivalent (in the case of locally bounded processes) to the existence of an equivalent local martingale measure (in Section \ref{S5} we will propose an example of a globally viable financial market that does not satisfy NFLVR).
We then have the following result\footnote{Actually, Theorem \ref{ex-defl} can be regarded as a special case of a more general result. Indeed, it has been recently proved in \cite{CDM} that, in a general semimartingale setting, utility maximization problems are solvable along an increasing sequence of stopping times if and only if NUPBR holds (and, hence, if and only if there exists a local martingale deflator, due to \cite{Kar}). However, except for the case of exponential utility, the general results of \cite{CDM} do not yield an explicit description of the local martingale deflator.}, the proof of which (postponed to the Appendix) is mainly based on elementary computations, relying on Assumption \ref{lcl-bdd}, together with the implication (i) $\Rightarrow$ (ii) of Theorem \ref{viability-PIEMM-loc}.

\begin{Thm}	\label{ex-defl}
If the financial market is globally viable, in the sense of Definition \ref{viability-glob}, then there exists a PILMD.
\end{Thm}

Note that, as for Remark \ref{duality}, the no-arbitrage property of the model (now in terms of NA1/NUPBR) comes as a  direct consequence of the (global) viability of the financial market. In our setting, however, the mere existence of a PILMD is not sufficient to ensure the global viability of the financial market. Nevertheless, it turns out that, if the financial market is globally viable, in the sense of Definition \ref{viability-glob}, and if the family of optimal strategies $\{\phi_n^*\}_{n\in\N}$ satisfies the following consistency condition, then the global viability of the financial market is equivalent to the existence of a PILMD which aggregates the density processes of all local PIEMMs (see Theorem \ref{viability-global}).

\begin{Def}	\label{cons-cond}
If the financial market is globally viable, in the sense of Definition \ref{viability-glob}, we say that the family of optimal strategies $\{\phi^*_n\}_{n\in\N}$ is \emph{consistent} if the following hold:
\be	\label{consistency}
\frac{E\bigl[U'\bigl(X_{\phi^*_n}(T\wedge\tau_n)\bigr)|\F_{T\wedge\tau_{n-1}}\bigr]}{E\bigl[U'\bigl(X_{\phi^*_n}(T\wedge\tau_n)\bigr)\bigr]} 
= \frac{U'\bigl(X_{\phi^*_{n-1}}(T\wedge\tau_{n-1})\bigr)}{E\bigl[U'\bigl(X_{\phi^*_{n-1}}(T\wedge\tau_{n-1})\bigr)\bigr]}\,,
\qquad\text{ for all }n\in\N.
\ee
\end{Def}

\begin{Rem}	\label{disc-cond}
\begin{myenumerate}
\item
As a first and easy case, condition \eqref{consistency} holds whenever the model is \emph{complete}, in the sense that, for every $n\in\N$, the set of all PIEMMs up to $\tau_n$ consists of a single element. Indeed, Theorem \ref{viability-PIEMM-loc} shows that $Q^{\phi^*_n}$ defined as in \eqref{thm-1} is a PIEMM up to $\tau_n$ and, analogously, $Q^{\phi^*_{n-1}}$ is a PIEMM up to $\tau_{n-1}$. However, since $\tau_{n-1}\leq\tau_n$ $P$-a.s. for every $n\in\N$, the measure $Q^{\phi^*_n}$ is also a PIEMM up to $\tau_{n-1}$ and, if there exists a unique PIEMM up to $\tau_{n-1}$, we then have $Q^{\phi^*_n}|_{\F_{T\wedge\tau_{n-1}}}=Q^{\phi^*_{n-1}}$. Due to \eqref{thm-1}, this clearly implies the consistency condition \eqref{consistency}.
\item 
As can be readily checked, a general sufficient condition for the validity of \eqref{consistency} is given by $\phi^*_{n-1}=\phi^*_n\ind_{\lsi0,\tau_{n-1}\rsi}$ together with the martingale property of $U'\bigl(X_{\phi^*_n}(\cdot\wedge\tau_n)\bigr)$, for every $n\in\N$. In particular, this holds in the following two cases: (a) for a logarithmic utility function in a wide class of jump-diffusion models, for which the reciprocal of the log-optimal portfolio is a local martingale (see e.g. \cite{KOS} as well as the example in Section \ref{S5}); (b) for an exponential utility function when trading strategies are restricted to a compact set, in a wide class of jump-diffusion models for which the measure $\nu$ is finite, as in \cite{M} (see also \cite{HIM} for related results in the continuous case)\footnote{Indeed, if the measure $\nu$ is finite and trading strategies are restricted to a compact set, the local boundedness of $b$, $\sigma$ and $\gamma$ together with the results of Section 4 of \cite{M} imply the martingale property of $U\bigl(X_{\phi^*_n}(\cdot\wedge\tau_n)\bigr)$ (it can be shown that the class of admissible strategies adopted in \cite{M} coincides with $\A_{\FF}^U$), for every $n\in\N$. In particular, since the utility is exponential, this implies the martingale property of $U'\bigl(X_{\phi^*_n}(\cdot\wedge\tau_n)\bigr)$ and also that $\phi^*_{n-1}=\phi^*_n\ind_{\lsi0,\tau_{n-1}\rsi}$, for every $n\in\N$.}.
\item
Note that, in view of Proposition \ref{BSDE-opt}, condition \eqref{consistency} can be equivalently formulated as a consistency condition on the terminal values of the solutions $\{p^n\}_{n\in\N}\subset\M^2$ to the adjoint BSDE \eqref{BSDE}. 
\end{myenumerate}
\end{Rem}

By relying on Theorem \ref{viability-PIEMM-loc}, we can now formulate the announced equivalence result, which characterizes global market viability, with a consistent family of optimal strategies, under partial information.

\begin{Thm}	\label{viability-global}
The following are equivalent:
\begin{myenumerate}
\item
the financial market is globally viable, in the sense of Definition \ref{viability-glob}, with a consistent family $\{\phi^*_n\}_{n\in\N}$ of optimal strategies;
\item
there exists a family of strategies $\{\phi^*_n\}_{n\in\N}$, with $\phi^*_n\in\A_{\FF}^U(n)$, for every $n\in\N$, such that the process $Z=\bigl\{Z(t):t\in[0,T]\bigr\}$ defined by
\be	\label{Z}
Z(t) := \ind_{\{t=0\}}+\sum_{k=1}^{\infty}\ind_{\{\tau_{k-1}<t\leq\tau_k\}}
\frac{E\bigl[U'\bigl(X_{\phi^*_k}(T\wedge\tau_k)\bigr)|\F_t\bigr]}{E\bigl[U'\bigl(X_{\phi^*_k}(T\wedge\tau_k)\bigr)\bigr]},
\qquad\text{for all }t\in[0,T],
\ee
is a PILMD satisfying $Z(T\wedge\tau_n)=U'\bigl(X_{\phi^*_n}(T\wedge\tau_n)\bigr)/E\bigl[U'\bigl(X_{\phi^*_n}(T\wedge\tau_n)\bigr)\bigr]$, for all $n\in\N$, with $\tau_0:=0$.
\end{myenumerate}
\end{Thm}
\begin{proof}
{\bf (i) $\Rightarrow$ (ii)}:
Let the strategy $\phi_n^*\in\A^U_{\FF}(n)$ solve Problem \ref{utility-loc}, for every $n\in\N$, and define the process $Z_{\phi_k^*}=\bigl\{Z_{\phi_k^*}(t);t\in[0,T]\bigr\}$ as in \eqref{Z^k}, for every $k\in\N$, and the process $Z=\bigl\{Z(t);t\in[0,T]\bigr\}$ by $Z(t):=\prod_{k=1}^{\infty}\frac{Z_{\phi_k^*}(t\wedge\tau_k)}{Z_{\phi_k^*}(t\wedge\tau_{k-1})}$. As shown in the proof of Theorem \ref{ex-defl}, the process $Z$ is a PILMD. 
Moreover, as can be easily checked, condition \eqref{consistency} implies that $Z_{\phi_k^*}(t\wedge\tau_{k-1})=Z_{\phi_{k-1}^*}(t\wedge\tau_{k-1})$ for all $t\in[0,T]$ and $k\in\N$, so that the process $Z$ can be equivalently rewritten as follows, for all $t\in[0,T]$:
\[
Z(t) = \ind_{\{t=0\}}+\sum_{k=1}^{\infty}\ind_{\{\tau_{k-1}<t\leq\tau_k\}}Z_{\phi_k^*}(t)
=\ind_{\{t=0\}}+\sum_{k=1}^{\infty}\ind_{\{\tau_{k-1}<t\leq\tau_k\}}
\frac{E\bigl[U'\bigl(X_{\phi_k^*}(T\wedge\tau_k)\bigr)|\F_t\bigr]}{E\bigl[U'\bigl(X_{\phi_k^*}(T\wedge\tau_k)\bigr)\bigr]}.
\]
To complete the proof of the implication (i) $\Rightarrow$ (ii), note that the consistency condition \eqref{consistency}, together with the tower property of conditional expectation, implies that, for all $n\in\N$:
\[	\ba
Z(T\wedge\tau_n)
&= \sum_{k=1}^n\ind_{\{\tau_{k-1}<T\leq\tau_k\}}\frac{U'\bigl(X_{\phi_k^*}(T\wedge\tau_k)\bigr)}{E\bigl[U'\bigl(X_{\phi_k^*}(T\wedge\tau_k)\bigr)\bigr]}
+\ind_{\{T>\tau_n\}}\frac{U'\bigl(X_{\phi_n^*}(T\wedge\tau_n)\bigr)}{E\bigl[U'\bigl(X_{\phi_n^*}(T\wedge\tau_n)\bigr)\bigr]}	\\
&= \sum_{k=1}^n\ind_{\{\tau_{k-1}<T\leq\tau_k\}}\frac{E\left[U'\bigl(X_{\phi_n^*}(T\wedge\tau_n)\bigr)|\F_{T\wedge\tau_k}\right]}{E\bigl[U'\bigl(X_{\phi_n^*}(T\wedge\tau_n)\bigr)\bigr]}
+\ind_{\{T>\tau_n\}}\frac{U'\bigl(X_{\phi_n^*}(T\wedge\tau_n)\bigr)}{E\bigl[U'\bigl(X_{\phi_n^*}(T\wedge\tau_n)\bigr)\bigr]}	\\
&= \frac{U'\bigl(X_{\phi_n^*}(T\wedge\tau_n)\bigr)}{E\bigl[U'\bigl(X_{\phi_n^*}(T\wedge\tau_n)\bigr)\bigr]}\,.
\ea	\]

\noindent {\bf (ii) $\Rightarrow$ (i)}:
Suppose that, for some family of strategies $\{\phi_n^*\}_{n\in\N}$ with $\phi_n^*\in\A_{\FF}^U(n)$, for all $n\in\N$, the process $Z$ defined in \eqref{Z} is a PILMD. Due to the local boundedness assumption (see part (i) of Assumption \ref{lcl-bdd}), the sequence $\{\tau_n\}_{n\in\N}$ is a localizing sequence for $Z$, meaning that the stopped process $Z^{\tau_n}$ is a $\FF$-martingale, for every $n\in\N$. Moreover, since $Z$ is a PILMD, the stopped process $(ZS)^{\tau_n}$ is an $\FF$-martingale, for every $n\in\N$.
In view of Definition \ref{PIEMM-loc}, this means that the measure $Q^n$ defined by $dQ^n:=Z(T\wedge\tau_n)dP$ is a PIEMM up to $\tau_n$. Theorem \ref{viability-PIEMM-loc} then implies that $\phi_n^*$ solves Problem \ref{utility-loc}, for all $n\in\N$. To complete the proof, it remains to prove the consistency condition \eqref{consistency}. This can be shown as follows, using the $\FF$-martingale property of $Z^{\tau_n}$, for all $n\in\N$ and $t\in[0,T]$:
\[
\frac{E\bigl[U'\bigl(X_{\phi_n^*}(T\wedge\tau_n)|\F_{T\wedge\tau_{n-1}}\bigr]}{E\bigl[U'\bigl(X_{\phi_n^*}(T\wedge\tau_n)\bigr]}
-\frac{U'\bigl(X_{\phi_{n-1}^*}(T\wedge\tau_{n-1})}{E\bigl[U'\bigl(X_{\phi_{n-1}^*}(T\wedge\tau_{n-1})\bigr]}
= E\bigl[Z_{\phi^*}(T\wedge\tau_n)|\F_{T\wedge\tau_{n-1}}\bigr]-Z_{\phi^*}(T\wedge\tau_{n-1}) = 0\,.
\]
In view of Definition \ref{viability-glob}, we have thus shown that the financial market is globally viable with a consistent family $\{\phi^*_n\}_{n\in\N}$ of optimal strategies.
\end{proof}

In particular, we want to remark that Theorem \ref{viability-global} gives an explicit description of the PILMD $Z$, which aggregates the expected (normalised) marginal utilities of terminal wealth at the stopping times of the localizing sequence $\{\tau_n\}_{n\in\N}$.

\section{An example}	\label{S5}

This section is meant to be an illustration of the concepts discussed so far in the context of a simple continuous model, based on a three-dimensional Bessel process. Three-dimensional Bessel processes have been extensively studied in relation with the existence of arbitrage opportunities, see e.g. \cite{DS}, Section 2 of \cite{FR} and Example 4.6 in \cite{KK}. In the present example, we will show that, for a logarithmic utility function, the financial market is viable in the local as well as in the global sense, even though the model allows for arbitrage opportunities.

Let $(\Omega,\G,\GG,P)$ be a given filtered probability space, with a standard Brownian motion $B$ and where $\FF:=\GG^B\subseteq\GG$ is the $P$-augmented filtration generated by $B$. We define the (discounted) price process $S$ of a single risky asset as the solution to the following SDE:
\be	\label{Bessel-1}	\left\{	\ba
dS(t) &= \frac{1}{S(t)}\,dt+dB(t),	\quad t\in[0,T];	\\
S(0) &= 1.
\ea	\right.	\ee
The solution to the SDE \eqref{Bessel-1} is a $P$-a.s. strictly positive process known as the three-dimensional Bessel process (see e.g. \cite{RY}, Chapter XI). It is easy to see that there exists a sequence $\{\tau_n\}_{n\in\N}$ of $\FF$-stopping times with $\tau_n\nearrow+\infty$ $P$-a.s. as $n\rightarrow+\infty$ such that $S^{\tau_n}$ and $1/S^{\tau_n}$ are $P$-a.s. uniformly bounded, for every $n\in\N$. Indeed, it suffices to define $\tau_n:=\inf\bigl\{t\in[0,T]:S(t)\notin(1/n,n)\bigr\}$, for $n\in\N$ (with the usual convention $\inf\emptyset=+\infty$). 
A simple application of It\^o's formula gives that $dS^{-1}(t)=-S^{-2}(t)\,dB(t)$, thus showing that $1/S$ is an $\FF$-local martingale or, equivalently, that the stopped process $1/S^{\tau_n}$ is an $\FF$-martingale, for all $n\in\N$. Furthermore, for any $\phi\in\A_{\FF}^U(n)$ and $n\in\N$:
\be	\label{Bessel-2}	\ba
d\left(\frac{X_{\phi}(t\wedge\tau_n)}{S(t\wedge\tau_n)}\right)
&= X_{\phi}(t\wedge\tau_n)\,d\frac{1}{S(t\wedge\tau_n)}
+\frac{\phi(t)}{S(t\wedge\tau_n)}\,dS(t\wedge\tau_n)
-\ind_{\{\tau_n>t\}}\frac{\phi(t)}{S^2(t\wedge\tau_n)}dt	\\
&= X_{\phi}(t\wedge\tau_n)\,d\frac{1}{S(t\wedge\tau_n)}+\frac{\phi(t)}{S(t\wedge\tau_n)}\,dB(t\wedge\tau_n)
\ea	\ee
thus showing that the stopped process $X_{\phi}^{\tau_n}/S^{\tau_n}$ is an $\FF$-martingale, for every $n\in\N$.

Let us consider the logarithmic utility function $U(x)=\log(x)$, with an initial endowment of $x=1$. 
Jensen's inequality together with the martingale property of $X_{\phi}^{\tau_n}/S^{\tau_n}$ gives
\[
E\Bigl[\log\bigl(X_{\phi}(T\wedge\tau_n)/S(T\wedge\tau_n)\bigr)\Bigr]
\leq\log\Bigl(E\bigl[X_{\phi}(T\wedge\tau_n)/S(T\wedge\tau_n)\bigr]\Bigr)=0
\]
meaning that $E\bigl[\log\bigl(X_{\phi}(T\wedge\tau_n)\bigr)\bigr]\leq E\bigl[\log\bigl(S(T\wedge\tau_n)\bigr)\bigr]$, for any $\phi\in\A_{\FF}^U(n)$. This shows that the optimal strategy for the logarithmic utility is a simple buy-and-hold position in the risky asset itself, i.e., $\phi^*_n=1\in\A^U_{\FF}(n)$ for all $n\in\N$. According to Definition \ref{viability-loc}, the financial market is locally viable up to $\tau_n$, for every $n\in\N$. 

We can also verify the local viability of the financial market by applying Theorem \ref{viability-PIEMM-loc}. Indeed, since the stopped process $1/S^{\tau_n}$ is a strictly positive $\FF$-martingale, we can define a probability measure $Q^n$ on $(\Omega,\G_T)$ by letting $dQ^n/dP:=1/S(T\wedge\tau_n)$. Due to Bayes' rule, it is easy to check that $Q^n$ is a PIEMM up to $\tau_n$, in the sense of Definition \ref{PIEMM-loc}. Since $U'(x)=1/x$, the implication (ii) $\Rightarrow$ (i) of Theorem \ref{viability-PIEMM-loc} implies that the financial market is viable up to $\tau_n$ and that $\phi^*_n=1\in\A^U_{\FF}(n)$ solves Problem \ref{utility-loc}, for every $n\in\N$.

Since the process $1/S$ is unbounded, Assumption \ref{ass-easy} fails to hold and, hence, we cannot rely on the approach presented in Section \ref{S4.1} to study the global viability of the financial market. More precisely, we can prove that the process $1/S$ cannot be used as the density process of a PIEMM on $[0,T]$, since $1/S$ is a strict local martingale, according to the terminology of \cite{ELY}, being a local martingale which fails to be a true martingale, so that $E[1/S(T)]<1$. 
Let us explain with some more details this phenomenon. We define the measure $Q^{\phi^*}$ as follows:
\[
\frac{dQ^{\phi^*}}{dP} 
:= \frac{U'\bigl(X_{\phi^*}(T)\bigr)}{E\bigl[U'\bigl(X_{\phi^*}(T)\bigr)\bigr]}
= \frac{1/S(T)}{E\bigl[1/S(T)\bigr]}.
\]
If $Q^{\phi^*}$ were a PIEMM, then its density process $G=\bigl\{G(t);t\in[0,T]\bigr\}$, with $dQ^{\phi^*}|_{\F_t}:=G(t)\,dP|_{\F_t}$ for all $t\in[0,T]$, would be an $\FF$-martingale admitting the following representation, as in \eqref{density-2}:
\[
G(t)=\exp\left(\int_0^t\!\theta_0(s)dB(s)-\frac{1}{2}\int_0^t\!\theta_0^2(s)ds\right),
\qquad\text{for all }t\in[0,T],
\]
for some $\FF$-predictable process $\theta_0=\bigl\{\theta_0(t);t\in[0,T]\bigr\}$ with $\int_0^T\!\theta_0^2(t)dt<\infty$ $P$-a.s., so that:
\be	\label{Bessel-3}
d\bigl(G(t)S(t)\bigr) = S(t)dG(t) + G(t)dS(t) + d\bigl\langle G,S\bigr\rangle(t)
= S(t)dG(t) + G(t)dB(t) + G(t)\!\left(\frac{1}{S(t)}+\theta_0(t)\right)dt.
\ee
If $Q^{\phi^*}$ were a PIEMM, then the product $GS$ would be an $\FF$-(local) martingale and equation \eqref{Bessel-3} would then imply that $\theta_0(t)=-1/S(t)$ for a.a. $t\in[0,T]$, meaning that $dG(t)=-G(t)/S(t)\,dB(t)$. But, since $G(0)=1/S(0)=1$, this would in turn imply that $G$ and $1/S$ solve the same SDE and, hence, one would conclude that $G=1/S$, thus contradicting the martingale property of $G$. This shows that, in the context of the present example, the marginal utility of the optimal terminal wealth cannot be taken as the density of a PIEMM. The failure of the martingale property of $1/S$ is also linked to the existence of multiple solutions to the BSDE \eqref{BSDE} on the time horizon $[0,T]$ beyond the class $\M^2\times L^2(B)$, as discussed in the recent paper \cite{X}.

We conclude the discussion of this example by showing that the financial market is globally viable with a consistent family of optimal strategies, in the sense of Definitions \ref{viability-glob} and \ref{cons-cond}. Indeed, we already know that $\phi^*:=1=\phi^*_n\in\A_{\FF}^U(n)$ solves Problem \ref{utility-loc}, for every $n\in\N$. Moreover, the consistency condition \eqref{consistency} also holds, due to the martingale property of the stopped process $1/S^{\tau_n}$, as explained in part (ii) of Remark \ref{disc-cond}.
Alternatively, one can prove the global viability of the financial market by applying Theorem \ref{viability-global}. Indeed,  take $\phi=1\in\bigcap_{n\in\N}\A_{\FF}^U(n)$ and consider the process $Z_{\phi}$ defined in \eqref{Z}. Since $U'(x)=1/x$, for every $n\in\N$, it is immediate to check that $Z_{\phi}=1/S$:
\[
Z_{\phi}(t) = \ind_{\{t=0\}}
+\sum_{k=1}^{\infty}\ind_{\{\tau_{k-1}<t\leq\tau_k\}}\frac{E\bigl[1/S(T\wedge\tau_k)|\F_t\bigr]}{E\bigl[1/S(T\wedge\tau_k)\bigr]}
= \ind_{\{t=0\}}+\sum_{k=1}^{\infty}\ind_{\{\tau_{k-1}<t\leq\tau_k\}}\frac{1}{S(t\wedge\tau_k)}
= \frac{1}{S(t)}.
\]
Equation \eqref{Bessel-2} together with the martingale property of $1/S^{\tau_n}$, for all $n\in\N$, shows that $1/S$ is a PILMD and Theorem \ref{viability-global} implies then that the financial market is globally viable with a consistent family of optimal strategies.

\vspace{1cm}

\setstretch{1}
\noindent \textbf{Acknowledgements:}
The authors are thankful to the anonymous referees for helpful comments that helped to improve the paper. 
The research leading to these results has received funding from the European Research Council under the European Community's Seventh Framework Programme (FP7/2007-2013) / ERC grant agreement no [228087].

\appendix
\sets

\section{Appendix}

\subsubsection*{Proof of Lemma \ref{BSDE-sol}.}

For any fixed $\phi\in\A^U_{\FF}(n)$, the random variable $U'\bigl(X_{\phi}(T\wedge\tau_n)\bigr)$ belongs to $L^2(P)$ and, hence, the $\GG$-martingale $p^n=\bigl\{p^n(t);t\in[0,T]\bigr\}$ defined by $p^n(t):=E\bigl[U'\bigl(X_{\phi}(T\wedge\tau_n)\bigr)|\G_t\bigr]$, for all $t\in[0,T]$, satisfies $p^n(T)=U'\bigl(X_{\phi}(T\wedge\tau_n)\bigr)$ and belongs to $\M^2$, as a consequence of Doob's inequality. 
Since the pair $(B,\tN)$ enjoys the martingale representation property in the filtration $\GG$ (see e.g. \cite{R}, Theorem 2.3) and since the martingale representation property is stable under stopping (see e.g. \cite{HWY}, Lemma 13.8), there exists a unique couple $(q^n,r^n)\in L^2(B)\times G^2(\tN)$ such that \eqref{BSDE} is satisfied.
In order to prove the integrability property of $p^n$, it suffices to note that, using Doob's inequality:
\[
E\left[\int_0^T\!\bigl(p^n(t)\bigr)^2dt\right] 
\leq T\,E\left[\underset{t\in[0,T]}{\sup}\bigl|p^n(t)\bigr|^2\right] 
\leq 4T\,E\left[\bigl(p^n(T)\bigr)^2\right] 
= 4T\,E\Bigl[U'\bigl(X_{\phi}(T\wedge\tau_n)\bigr)^2\Bigr] < \infty.
\]

\subsubsection*{Proof of Proposition \ref{BSDE-opt}.}

The proof of the proposition follows from an adaptation of the sufficient and necessary maximum principles under partial information proven in \cite{BO}, together with a suitable localization procedure. However, for the convenience of the reader, we prefer to give full details.
Suppose that, for some $\phi\in\A_{\FF}^U(n)$, the unique solution $(p^n,q^n,r^n)\in\M^2\times L^2(B)\times G^2(\tN)$ to the BSDE \eqref{BSDE} satisfies condition \eqref{BSDE-cond} and let $\bar{\phi}$ be any element of $\A_{\FF}^U(n)$. Define the sequence of stopping times $\{\tilde{\varrho}_k\}_{k\in\N}$ by
\[
\tilde{\varrho}_k := \inf\bigl\{t\in[0,T]:\left|X_{\phi}(t)-X_{\bar{\phi}}(t)\right|\geq k\text{ or }\left|p^n(t)\right|\geq k\bigr\}
\]
and let $\varrho_{k,n}:=\tau_n\wedge\tilde{\varrho}_k$, so that $\lim_{k\rightarrow\infty}\varrho_{k,n}=\tau_n$ $P$-a.s. for $k\rightarrow\infty$. 
Then, using the concavity of $U$:
\begin{equation}	\label{A1} \begin{aligned}
E\Bigl[U\bigl(X_{\bar{\phi}}(T\wedge\tau_n)\bigr)-U\bigl(X_{\phi}(T\wedge\tau_n)\bigr)\Bigr]
&\leq E\Bigl[U'\bigl(X_{\phi}(T\wedge\tau_n)\bigr)\bigl(X_{\bar{\phi}}(T\wedge\tau_n)-X_{\phi}(T\wedge\tau_n)\bigr)\Bigr]	\\
&= E\Bigl[p^n(T)\bigl(X_{\bar{\phi}}(T\wedge\tau_n)-X_{\phi}(T\wedge\tau_n)\bigr)\Bigr]	\\
&= E\Bigl[\,\underset{k\rightarrow\infty}{\lim}\;
p^n(T\wedge\varrho_{k,n})\bigl(X_{\bar{\phi}}(T\wedge\varrho_{k,n})-X_{\phi}(T\wedge\varrho_{k,n})\bigr)\Bigr]\\
&= \underset{k\rightarrow\infty}{\lim}E\Bigl[p^n(T\wedge\varrho_{k,n})\bigl(X_{\bar{\phi}}(T\wedge\varrho_{k,n})-X_{\phi}(T\wedge\varrho_{k,n})\bigr)\Bigr],
\end{aligned}	\end{equation}
where the last equality follows from the dominated convergence theorem, since the random variable $\sup_{t\in[0,T\wedge\tau_n]}\left|p^n(t)\bigl(X_{\bar{\phi}}(t)-X_{\phi}(t)\bigr)\right|$ is integrable, as a consequence of the Cauchy-Schwarz inequality together with the fact that $p^n\in\M^2$ and $\phi,\bar{\phi}\in\A_{\FF}^U(n)$.
Then, by applying the product rule:
\begin{equation}	\label{A2}	\begin{aligned}
& p^n(T\wedge\varrho_{k,n})\bigl(X_{\bar{\phi}}(T\wedge\varrho_{k,n})-X_{\phi}(T\wedge\varrho_{k,n})\bigr)	\\
&\quad= \int_0^{T\wedge\varrho_{k,n}}\!\bigl(X_{\bar{\phi}}(t-)-X_{\phi}(t-)\bigr)\,dp^n(t)
+\int_0^{T\wedge\varrho_{k,n}}\!p^n(t-)\,d\bigl(X_{\bar{\phi}}(t)-X_{\phi}(t)\bigr)
+\Bigl[p^n,X_{\bar{\phi}}-X_{\phi}\Bigr](T\wedge\varrho_{k,n})\,.
\end{aligned}	\end{equation}
Note that the first stochastic integral in \eqref{A2} belongs to $\M^2$, since $\bigl|X_{\bar{\phi}}(t-)-X_{\phi}(t-)\bigr|\leq k$ $P$-a.s. for $t\leq\varrho_{k,n}$ and $p^n\in\M^2$. Analogously, also the martingale part of the second stochastic integral in \eqref{A2} belongs to $\M^2$, since $|p^n(t-)|\leq k$ $P$-a.s. for $t\leq\varrho_{k,n}$ and $\bigl(X_{\bar{\phi}}-X_{\phi}\bigr)\in\calS^2$, due to the definition of $\A_{\FF}^U(n)$. Finally, we also have $\bigl[p^n,X_{\bar{\phi}}-X_{\phi}\bigr]-\bigl\langle p^n,X_{\bar{\phi}}-X_{\phi}\bigr\rangle\in\M^2$ (see e.g. \cite{HWY}, Section VI.4).
These observations allow us to write:
\begin{equation}	\label{A3}	\begin{aligned}
& E\Bigl[p^n(T\wedge\varrho_{k,n})\bigl(X_{\bar{\phi}}(T\wedge\varrho_{k,n})-X_{\phi}(T\wedge\varrho_{k,n})\bigr)\Bigr]	\\
&\quad= E\left[\int_0^{T\wedge\varrho_{k,n}}\!\bigl(\bar{\phi}(t)-\phi(t)\bigr)
\Bigl(p^n(t)b^n(t)+\sigma^n(t)q^n(t)+\int_{\R}\gamma^n(t,\zeta)r^n(t,\zeta)\nu(d\zeta)\Bigr)dt\right]\,.
\end{aligned}	\end{equation}
Hence, together with \eqref{A1} and \eqref{A3}, letting $k\rightarrow\infty$ and applying again the dominated convergence theorem (which is justified since $p^n$ satisfies $E\left[\int_0^T\!\bigl(p^n(t)\bigr)^2dt\right]<\infty$, $(q^n,r^n)\in L^2(B)\times G^2(\tN)$ and $\bar{\phi},\phi\in\A_{\FF}^U(n)$) lead to:
\begin{equation}	\label{A4}	\begin{aligned}
& E\Bigl[U\bigl(X_{\bar{\phi}}(T\wedge\tau_n)\bigr)-U\bigl(X_{\phi}(T\wedge\tau_n)\bigr)\Bigr]	\\
&\quad\leq E\left[\int_0^{T\wedge\tau_n}\!\bigl(\bar{\phi}(t)-\phi(t)\bigr)
\Bigl(p^n(t)b^n(t)+\sigma^n(t)q^n(t)+\int_{\R}\gamma^n(t,\zeta)r^n(t,\zeta)\nu(d\zeta)\Bigr)dt\right]	\\
&\quad= E\left[\int_0^{T\wedge\tau_n}\!\bigl(\bar{\phi}(t)-\phi(t)\bigr)
E\Bigl[p^n(t)b^n(t)+\sigma^n(t)q^n(t)+\int_{\R}\gamma^n(t,\zeta)r^n(t,\zeta)\nu(d\zeta)\bigr|\F_t\Bigr]dt\right] = 0	\,,
\end{aligned}	\end{equation}
where we have used the fact that $\{t<\tau_n\}\in\F_t$, since $\tau_n$ is an $\FF$-stopping time (see part (i) of Assumption \ref{lcl-bdd}), and that the strategies $\bar{\phi},\phi$ are $\cP_{\FF}$-measurable, and where the last equality is due to condition \eqref{BSDE-cond}.	This proves that $\phi\in\A_{\FF}^U(n)$ solves Problem \eqref{utility-loc}.	\\
Conversely, let $\phi\in\A_{\FF}^U(n)$ solve Problem \eqref{utility-loc}. Proceeding as in the proof of Theorem 3.1 of \cite{BO}, define $h(\delta):=E\left[U\bigl(X_{\phi+\delta\beta}(T\wedge\tau_n)\bigr)\right]$, where  $\beta=\bigl\{\beta(t);t\in[0,T]\bigr\}$ is a $\cP_{\FF}$-measurable process with $0\leq\beta(t)\leq 1$ $P$-a.s. for all $t\in[0,T]$ such that $\phi+\delta\beta\in\A_{\FF}^U(n)$, for every $\delta\in(-\bar{\delta},\bar{\delta}\,)$, for some constant $\bar{\delta}>0$. As can be readily verified, the optimality of $\phi\in\A_{\FF}^U(n)$ implies that
\[
0 = h'(0)
= E\left[U'\bigl(X_{\phi}(T\wedge\tau_n)\bigr)X_{\beta}(T\wedge\tau_n)\right]
\]
and, since $S^{\tau_n}$ is $P$-a.s. uniformly bounded\footnote{Note that the boundedness of $S^{\tau_n}$ was not used so far in this proof. In particular, this implies that condition \eqref{BSDE-cond} can act as a sufficient condition for optimality on the global investment horizon $[0,T]$ even when $S$ is not locally bounded.}, the same arguments used in \eqref{A1}-\eqref{A4} allow us to write
\begin{equation}	\label{A5}	\begin{aligned}
& E\left[\int_0^{T\wedge\tau_n}\!\beta(t)\Bigl(p^n(t)b^n(t)+\sigma^n(t)q^n(t)+\int_{\R}\gamma^n(t,\zeta)r^n(t,\zeta)\nu(d\zeta)\Bigr)dt\right]	\\
&\quad= E\left[U'\bigl(X_{\phi}(T\wedge\tau_n)\bigr)X_{\beta}(T\wedge\tau_n)\right] = 0\,.
\end{aligned}	\end{equation}
Take then a process $\beta$ of the form $\beta(t):=\ind_A\ind_{[t_0,t_0+h]}(t)$, for some $t_0<T$ and $h>0$, with $t_0+h\leq T$, and where $A\in\F_{t_0}$. Since $S^{\tau_n}$ is $P$-a.s. uniformly bounded, it is clear that $\beta$ satisfies the above assumptions and, moreover, with this choice of $\beta$, equation \eqref{A5} leads to
\[
E\left[\int_{t_0\wedge\tau_n}^{(t_0+h)\wedge\tau_n}\!\ind_A\Bigl(p^n(t)b^n(t)+\sigma^n(t)q^n(t)+\int_{\R}\gamma^n(t,\zeta)r^n(t,\zeta)\nu(d\zeta)\Bigr)dt\right] = 0
\]
and differentiating with respect to $h$ at $h=0$ gives
\[
E\left[\ind_{\{t_0<\tau_n\}}\ind_A\Bigl(p^n(t_0)b^n(t_0)+\sigma^n(t_0)q^n(t_0)+\int_{\R}\gamma^n(t_0,\zeta)r^n(t_0,\zeta)\nu(d\zeta)\Bigr)\right] = 0\,.
\]
Since $t_0$ and $A\in\F_{t_0}$ are arbitrary and $\tau_n$ is an $\FF$-stopping time, we have thus proven that \eqref{BSDE-cond} holds.

\subsubsection*{Proof of Proposition \ref{PIEMM}.}

Due to Assumption \ref{lcl-bdd}, the stopped process $S^{\tau_n}$ is $\FF$-adapted. Hence, the conditional version of Bayes' rule gives, for all $0\leq s\leq t\leq T$
\begin{equation}	\label{proof-1}
E_{Q_{\theta}}\bigl[S^{\tau_n}(t)|\F_s\bigr]-S^{\tau_n}(s)
= \frac{E\bigl[G_{\theta}(t)S(t\wedge\tau_n)
-G_{\theta}(s)S(s\wedge\tau_n)\bigr|\F_s\bigr]}{E[G_{\theta}(s)|\F_s]}.
\end{equation}
Furthermore, by applying the integration by parts formula (see e.g. \cite{OS1}, Lemma 3.6)
\[ \ba
d\bigl(G_\theta(t)S(t\wedge\tau_n)\bigr) 
&= G_\theta(t-)\,dS(t\wedge\tau_n)+S(t\wedge\tau_n-)\,dG_\theta(t)+d\bigl[G_\theta,S\bigr](t\wedge\tau_n) \\
&= G_\theta(t-)\,\ind_{\{t\leq\tau_n\}}\!
\left(b(t)dt+\sigma(t)dB(t)+\int_\R\!\gamma(t,\zeta)\tN(dt,d\zeta)\right) \\
&\quad +S(t\wedge\tau_n-)
\left(G_\theta(t-)\Bigl(\theta_0(t)dB(t)+\int_\R\!\theta_1(t,\zeta)\tN(dt,d\zeta)\Bigr)\right)\\
&\quad +\ind_{\{t\leq\tau_n\}}G_\theta(t)\sigma(t)\theta_0(t)dt
+\ind_{\{t\leq\tau_n\}}\!\int_\R\!G_\theta(t-)\gamma(t,\zeta)\theta_1(t,\zeta)N(dt,d\zeta).
\ea	\]
Collecting the $dt$-terms we get, for all $t\in[0,T]$
\be	\label{proof-2}
G_\theta(t)S(t\wedge\tau_n) = S_0+\int_0^{t\wedge\tau_n}\!\!G_\theta(u)\!\left(b(u)+\sigma(u)\theta_0(u)
+\!\int_\R\!\gamma(u,\zeta)\theta_1(u,\zeta)\nu(d\zeta)\right)\!du+(\text{$\GG$-local martingale}).
\ee
Since $(\theta_0,\theta_1)\in L^2(B)\times G^2(\tN)$ and $S^{\tau_n}$, $b^n$, $\sigma^n$ and $\gamma^n$ are bounded, it can be easily verified that the $\GG$-local martingale term appearing in \eqref{proof-2} is actually a true $\GG$-martingale and, hence, we can write
\begin{align}
&E\bigl[G_{\theta}(t)S(t\wedge\tau_n)-G_{\theta}(s)S(s\wedge\tau_n)\bigr|\F_s\bigr]	
= E\Bigl[E\bigl[G_{\theta}(t)S(t\wedge\tau_n)-G_{\theta}(s)S(s\wedge\tau_n)|\G_s\bigr]\bigr|\F_s\Bigr]	\nonumber\\
&= E\left[\int_s^t\!\ind_{\{u<\tau_n\}}G_\theta(u)\left(b(u)+\sigma(u)\theta_0(u)
+\int_\R\!\gamma(u,\zeta)\theta_1(u,\zeta)\nu(d\zeta)\right)\!du\Bigr|\F_s\right].
\label{proof-3}
\end{align}
In view of equation \eqref{proof-1}, this shows that $S^{\tau_n}$ is a $(Q_{\theta},\FF)$-martingale if and only if \eqref{proof-3} is $P$-a.s. equal to zero for all $s,t\in[0,T]$ with $s\leq t$. Since $\tau_n$ is an $\FF$-stopping time (see Assumption \ref{lcl-bdd}), this is equivalent to the validity of condition \eqref{PIEMM-cond}.

\subsubsection*{Proof of Lemma \ref{moments}.}

We use Minkowski's inequality and the Burkholder-Davis-Gundy inequality to get, for any $r\in(1,\infty)$
\begin{gather*}
E\left[\underset{t\in[0,T]}{\sup}\left|\int_0^t\!\sigma(u)dB(u)
+\int_0^t\!\int_\R\gamma(u,\zeta)\tN(du,d\zeta)\right|^r\right]^{1/r}	\\
\leq E\left[\underset{t\in[0,T]}{\sup}\left|\int_0^t\!\sigma(u)dB(u)\right|^r\right]^{1/r}
+E\left[\underset{t\in[0,T]}{\sup}\left|\int_0^t\!\int_\R\gamma(u,\zeta)\tN(du,d\zeta)\right|^r\right]^{1/r}	\\
\leq C E\left[\left(\int_0^T\!\!\sigma^2(u)du\right)^{r/2}\right]^{1/r}
+CE\left[\left(\int_0^T\!\!\int_{\R}\gamma^2(u,\zeta)N(du,d\zeta)\right)^{r/2}\right]^{1/r}
<\infty
\end{gather*}
where $C$ is a positive constant and the finiteness of the last expectations follows from the uniform boundedness of $\sigma$ and $\gamma$ together with the finiteness of the measure $\nu$. Due to \eqref{S} and to the boundedness of $b$, this suffices to prove the claim.

\subsubsection*{Proof of Theorem \ref{ex-defl}.}

Let the financial market be globally viable, with a corresponding family of optimal strategies $\{\phi^*_n\}_{n\in\N}$, with $\phi^*_n\in\A_{\FF}^U(n)$, for all $n\in\N$. For every $k\in\N$, define the process $Z_{\phi_k^*}=\bigl\{Z_{\phi_k^*}(t);t\in[0,T]\bigr\}$ by
\be	\label{Z^k}
Z_{\phi_k^*}(t) := \frac{E\bigl[U'\bigl(X_{\phi^*_k}(T\wedge\tau_k)\bigr)|\F_t\bigr]}
{E\bigl[U'\bigl(X_{\phi^*_k}(T\wedge\tau_k)\bigr)\bigr]},
\qquad\text{for all }t\in[0,T]\,.
\ee
Define then the $\FF$-adapted process $Z=\bigl\{Z(t);t\in[0,T]\bigr\}$ by $Z(t):=\prod_{k=1}^{\infty}\frac{Z_{\phi^*_k}(t\wedge\tau_k)}{Z_{\phi^*_k}(t\wedge\tau_{k-1})}$, for all $t\in[0,T]$, with $\tau_0:=0$. Since $U$ is strictly increasing, the process $Z$ is strictly positive and satisfies $Z(0)=1$. Furthermore, for every $n\in\N$ and $s,t\in[0,T]$ with $s\leq t$, using the $\FF$-martingale property of every $Z_{\phi^*_k}$ together with the fact that $\tau_n$ is an $\FF$-stopping time, we get
\be	\label{defl-comp}	\begin{aligned}
E\bigl[Z(t\wedge\tau_n)|\F_s\bigr]
&= E\bigl[\ind_{\{\tau_n\leq s\}}Z(t\wedge\tau_n)+\ind_{\{\tau_n>s\}}Z(t\wedge\tau_n)|\F_s\bigr]	\\
&= \ind_{\{\tau_n\leq s\}}Z(\tau_n)+\ind_{\{\tau_n>s\}}E\bigl[Z(t\wedge\tau_n)|\F_s\bigr]	\\
&= \ind_{\{\tau_n\leq s\}}Z(\tau_n)
+\sum_{\ell=1}^n\ind_{\{\tau_{\ell-1}\leq s<\tau_{\ell}\}}
E\left[Z(t\wedge\tau_n)|\F_s\right]	\\
&= \ind_{\{\tau_n\leq s\}}Z(\tau_n)
+\sum_{\ell=1}^n\ind_{\{\tau_{\ell-1}\leq s<\tau_{\ell}\}}
\prod_{k=1}^{\ell-1}\frac{Z_{\phi^*_k}(t\wedge\tau_k)}{Z_{\phi^*_k}(t\wedge\tau_{k-1})}
E\left[\,\prod_{k=\ell}^n\frac{Z_{\phi^*_k}(t\wedge\tau_k)}{Z_{\phi^*_k}(t\wedge\tau_{k-1})}\Bigr|\F_s\right]	\\
&= \ind_{\{\tau_n\leq s\}}Z(\tau_n)
+\sum_{\ell=1}^n\ind_{\{\tau_{\ell-1}\leq s<\tau_{\ell}\}}
\prod_{k=1}^{\ell-1}\frac{Z_{\phi^*_k}(s\wedge\tau_k)}{Z_{\phi^*_k}(s\wedge\tau_{k-1})}
\frac{E\bigl[Z_{\phi^*_{\ell}}(t\wedge\tau_{\ell})|\F_s\bigr]}{Z_{\phi^*_{\ell}}(\tau_{\ell-1})}	\\
&= \ind_{\{\tau_n\leq s\}}Z(\tau_n)
+\sum_{\ell=1}^n\ind_{\{\tau_{\ell-1}\leq s<\tau_{\ell}\}}
\prod_{k=1}^{\ell}\frac{Z_{\phi^*_k}(s\wedge\tau_k)}{Z_{\phi^*_k}(s\wedge\tau_{k-1})}	\\
&= \ind_{\{\tau_n\leq s\}}Z(\tau_n)+\ind_{\{\tau_n>s\}}Z(s)
= Z(s\wedge\tau_n)
\end{aligned}	\ee
where the fifth equality follows from a repeated application of the tower property of conditional expectation.
We have thus shown that, for every $n\in\N$, the process $Z(\cdot\wedge\tau_n)$ is an $\FF$-martingale. Since the sequence $\{\tau_n\}_{n\in\N}$ is composed of $\FF$-stopping times (see part (i) of Assumption \ref{lcl-bdd}), this implies the $\FF$-local martingale property of $Z$. 
In order to show that $Z$ is a PILMD, it remains to prove the $\FF$-local martingale property of $ZS$. Noting that $Z(t)S(t)=S(0)\prod_{k=1}^{\infty}\frac{S(t\wedge\tau_k)}{S(t\wedge\tau_{k-1})}\frac{Z_{\phi^*_k}(t\wedge\tau_k)}{Z_{\phi^*_k}(t\wedge\tau_{k-1})}$, the latter can be established by means of computations analogous to \eqref{defl-comp}, using the fact that $Z_{\phi^*_k}\,S^{\tau_k}$ is an $\FF$-martingale, for every $k\in\N$, as a consequence of the global viability of the financial market (see Definition \ref{viability-glob}) together with the implication (i) $\Rightarrow$ (ii) of Theorem \ref{viability-PIEMM-loc}.

\vspace{1cm}

\setstretch{1.15}

\end{document}